\documentclass{article}
\pdfoutput=1
\usepackage{arxiv,amsthm,amssymb}

\usepackage[utf8]{inputenc} 
\usepackage[T1]{fontenc}    
\usepackage{hyperref}       
\usepackage{url}            
\usepackage{booktabs}       
\usepackage{amsfonts}       
\usepackage{nicefrac}       
\usepackage{microtype}      
\usepackage{lipsum}		
\usepackage{graphicx}
\usepackage{doi}

\newtheorem{example}{Example}

\newtheorem{theorem}{Theorem}
\newtheorem{lemma}{Lemma}

\newcommand{\Real}{\mathbb{R}}

\title{Dispersal-induced growth in a time-periodic environment}

\author{Guy Katriel\\ Department of Applied Mathematics, ORT Braude College,\\ Karmiel, Israel\\}

\date{}

\begin{document}

\maketitle

\begin{abstract}
Dispersal-induced growth (DIG) occurs when two populations with time-varying growth rates, each of which, when isolated, would become extinct, are able to persist and grow exponentially 
when dispersal among the two populations is present. 
This work provides a mathematical exploration of this surprising phenomenon, in the context of a deterministic model with periodic variation of growth rates, and characterizes the factors which 
are important in generating the DIG effect, and  the corresponding conditions
on the parameters involved.
\end{abstract}


\section{Introduction}

Exploring how the dispersal of organisms interacts with
environmental heterogeneity, both spatial and temporal, to determine population 
growth, is a central theme in ecological theory,
with important implications for environmental management and conservation \cite{baguette,cousens,hanski,lewis}.
Many plant and animal populations inhabit 
separate patches of varying size and quality, which are inter-connected by dispersal. 
A patch is called a {\it{source}} if it can sustain a population, and a {\it{sink}} if it is 
of such low quality that a population would not persist on it, if isolated. 
A basic  insight of source-sink theory is that populations in sinks may be sustained, and even exhibit 
positive growth rates, as a result of immigration from source patches \cite{dias,kawecki,pulliam}. A
more surprising phenomenon is that of Dispersal-induced Growth (DIG), whereby it is 
possible for populations in a set of patches, with dispersal among them, to persist and grow {\it{despite}} the fact that {\it{all}} these patches are sinks. This counter-intuitive effect was first explicitly discussed, in different frameworks, in \cite{jansen,roy}.  \cite{jansen} used a simplified `well-mixing'
model with stochastic environment, and occurrence of the DIG effect was derived.
\cite{roy},
who used the term `inflationary effect' for what we here call dispersal-induced growth, 
modelled direct dispersal among patches, with stochastic growth rates which are temporally  {\it{positively autocorrelated}}, and the possibility of the DIG effect was derived by heuristic arguments and demonstrated by extensive numerical simulations. 
\cite{matthews} experimentally confirmed the
DIG phenomenon in a laboratory system using {\it{Paramecium aurelia}}. See also \cite{cheong,williams} for surveys and discussions of `paradoxical' effects in population biology, in which coupling of losing strategies can
lead to persistence and growth. 

In the continuous-time deterministic context, the DIG phenomenon was discussed and numerically demonstrated by \cite{klausmeier}, using a simple two-patch model leading to a pair of ordinary differential 
equations, with periodic growth rates - see equations (\ref{eq1}),(\ref{eq2}) below.
The present work is devoted to the mathematical analysis of this model and its generalization to multiple patches.
Periodically varying growth rates can either be thought of as a proxy for auto-correlated environmental fluctuations as assumed in the stochastic model  of \cite{roy}, or they can model seasonal variations in the quality of patches, another ubiquitous ecological mechanism \cite{white}. The DIG phenomenon is manifested when, despite the fact that the 
time-averaged growth rate in each patch is negative, which would lead to extinction in each patch if it were isolated, dispersal among the   patches allows the populations to persist and grow (see Figure \ref{exinf} below). Such persistence may be 
desirable in the context of species conservation, or undesirable, as in the case of invasive species or pathogens.
In the recent work \cite{kortessis} the same model is obtained 
as a linearization of an SIR epidemic model, which is used to numerically demonstrate that epidemic control through non-pharmaceutical interventions may be hampered by the fact that control measures are applied in a non-synchronized manner in two regions which are inter-connected by flows of infective individuals, leading to persistence of a pathogen that would have been eradicated if movement between the two regions had been curtailed, or if the control measures had been synchronized among the regions.

A better understanding of the DIG 
phenomenon, beyond numerical simulations, 
requires mathematical analysis of relevant models.
For stochastically varying growth,
several researchers have obtained analytical results regarding the DIG effect. In the special case
of `well-mixed' systems in which dispersing individuals join a common pool from which they disperse to all patches, a simple and elegant analytical treatment is available \cite{bascompte,jansen,metz}.
The case of direct and limited dispersal among patches is considerably more 
difficult. \cite{morita} analyze
a discrete-time two-patch model, showing that the ratio of populations
in the two patches converges to a stationary distribution characterized 
by a Perron-Frobenius equation which can be solved numerically, and
in terms of which the total growth rate of the populations can be computed. \cite{evans} analyze a continuous-time 
stochastic model, obtaining an explicit expression for the stationary distribution in the two-patch case, from which the population growth rate
can be calculated, and providing explicit conditions for the 
occurrence of dispersal-induced growth. In the above-mentioned analytical works, it is assumed that
there is no auto-correlation in the environmental variation. 
\cite{schreiber} obtains an analytical approximation of the 
population growth rate for the case of many patches
in a discrete-time model including auto-correlation of the 
time-dependent growth rates, which shows that positive autocorrelation enhances the total population growth rate and 
thus the possibility of DIG, while spatial correlation among patches reduces this effect, in agreement with the 
numerical findings of \cite{roy}.

In the context of time-periodic, rather than stochastic environmental variation, \cite{bansaye} analyze
discrete-time models, in which the environment 
varies periodically between two states, and conditions
for occurrence of DIG are given. Recent work of \cite{benaim}
studies the DIG phenonemonon in a two-patch model in which the growth rates switch between two values, both in a periodic and stochastically.

In recent years there has been significant progress in 
analyzing models with periodic environmental variation in 
the {\it{continuous space}} case, in which movement of organisms is modeled by diffusion. In particular, \cite{liu19} have proved an important result on the 
monotonicity of the principal eigenvalue of such problems 
in dependence on the frequency of environmental forcing (see also \cite{liu22,Su}).
The recent work \cite{liu} has proved an analog of the 
monotonicity result of \cite{liu19} in the case of 
discrete patches, and also analyzes the asymptotics 
of the principal eigenvalues in the limits of low and 
of high frequency. 

In this work we will employ the results
of \cite{liu} to provide a detailed analysis of the DIG 
phenomena and the conditions under which it arises.
We show that occurrence of 
the DIG effect depends on an appropriate balance of {\it{three}} factors:
\begin{itemize}
	\item[(i)] Difference in the time-dependence of the growth rates in the different patches: the DIG
	effect {\it{cannot}} occur when all time-dependent growth rates are identical or sufficiently similar.
	
	\item[(ii)]	 Frequency of the variation in 
	growth rates: this frequency must be sufficiently {\it{small}} for 
	DIG to occur.
	
	\item[(iii)] Rate of dispersal: for DIG to occur, dispersal must be neither too 
	weak nor too strong. This was also the case in the stochastic simulations reported in \cite{roy}, and the analytical results in \cite{schreiber}. Thus, dispersal is a `double-edged sword' \cite{abbott,hudson} - its positive effect
	on population growth is supressed if its level is too high.
\end{itemize}
The precise formulations of our results are given in 
Section \ref{mainres}. The results are illustrated by means 
of numerical simulations and computations of the ranges of parameters for which DIG occurs in specific examples.
Preliminaries to the proof of the main theorems 
are provided in Section \ref{growth_rate}, in which we characterize the growth rates of the species in several ways.
In Section \ref{low_high} we we analyze the
behavior of the growth rate in two asymptotic regimes - the low and the high frequency limits, using the results of \cite{liu}. This analysis forms 
the basis for the proofs of our main theorems, given in Section
\ref{mainproof}.

Our results will explain central features of the observations made using numerical simulations, and 
in particular show that DIG is a robust 
phenomenon occuring for general periodic growth-rate profiles, as long as the parameters involved are in appropriate ranges. 
Some questions for further research are proposed in the Discussion.

\section{The main results}
\label{mainres}

\subsection{The model}

We consider populations of sizes $x_i(t)$ ($1\leq i\leq N$), inhabiting $N$ patches, and subject to time-periodic local growth rates $r_i(\omega t)$ ($1\leq i\leq N$), where it is assumed that $r_i(\theta)$ are $2\pi$-periodic 
functions, so that $r_i(\omega t)$ are periodic with 
period $T=\frac{2\pi}{\omega}$. We also assume dispersal 
among the patches $i,j$ ($i\neq j$) at rate $m\cdot L_{ij}$, where the the parameter $m\geq 0$ is used to control the dispersal rate of the species described, and 
the numbers $L_{ij}=L_{ji}\geq 0$ ($i\neq j$) encode the topology of the dispersal network and the relative rates of dispersal among different patches. We then have the differential equations
\begin{equation}\label{eqsys}
	x_i'=r_i(\omega t)x_i+m\sum_{j\neq i}L_{ij}(x_j-x_i),\;\;\;1\leq i\leq N. 
\end{equation}
Equivalently, defining the diagonal elements of the matrix $L$ by
\begin{equation}\label{diagonal}L_{ii}=-\sum_{j\neq i} L_{ij},\;\;\;1\leq j\leq N,\end{equation}
and setting
$$\bold{x}(t)=\left(\begin{array}{c}
	x_1(t)	\\
	\vdots	\\
	x_N(t)
\end{array} \right),\;\;\;\;R(\theta)=\left(\begin{array}{ccc}
	r_1(\theta)& 0&  0\\
	0 &  \ddots& 0 \\
	0 & 0 & r_N(\theta)
\end{array} \right),$$
the system (\ref{eqsys}) can be written as
\begin{equation}\label{eqsys1}
	\bold{x}'=[R(\omega t)+mL]\bold{x}.
\end{equation}
In addition to the assumptions that $L$ is 
symmetric, has non-negative non-diagonal elements, and
(\ref{diagonal}), we also make the standing assumption that $L$ is {\it{irreducible}}, which means that the dispersal network among the patches is {\it{connected}} (any two patches are connected by a path).

\begin{example}\label{ex1}
	The simplest two-patch case ($N=2$) is already of much interest, and it allows to obtain some explicit formulas which cannot be obtained in the general case (see examples \ref{ex2},\ref{ex3} below, and Section \ref{reduction}). In this case 
	$L_{12}=L_{21}=1$, $L_{11}=L_{22}=-1$, and the system (\ref{eqsys}) is
	\begin{equation}\label{eq1}x_1'= r_1(\omega t)x_1 +m (x_2- x_1),\end{equation}
	\begin{equation}\label{eq2}x_2'= r_2(\omega t)x_2 +m (x_1- x_2).\end{equation}
\end{example}
The model studied is linear, as it 
does not take into account density-dependent 
effects. However, the study of this system 
is also directly relevant to the understanding of more elaborate models including nonlinearity, since persistence of populations 
in such models depends on the behavior of the system obtained by linearization around the trivial equilibrium $\bold{x}=0$, which brings us back to the (\ref{eqsys1}). Therefore our results entail the occurrence of the DIG effect in nonlinear models (e.g. the epidemic model of \cite{kortessis}).

Any solution of (\ref{eqsys1}) with 
$x_i(0)>0$ ($1\leq i\leq N$) satisfies 
$x_i(t)>0$ for all $t>0$ (see Section \ref{principal}).
Given a function $x:[0,\infty)\rightarrow (0,\infty)$ we will
denote its growth rate (Lyapunov exponent) by
$$\Lambda[x]=\lim_{t\rightarrow \infty}\frac{1}{t}\ln(x(t)),$$
provided this limit exists. Note that $\Lambda[x]>0$ 
corresponds to exponential growth, while $\Lambda[x]<0$ 
corresponds to exponential decay - leading to extinction. Therefore our investigation focuses 
on the quantities $\Lambda[x_i]$.

In the absence of dispersal ($m=0$) the population in each patch would evolve independently, and the differential equations are easily solved to yield
\begin{equation}\label{dc}x_i(t)=x_i(0)e^{\int_0^t r_i(\omega \tau)d\tau},\;\;1\leq i\leq N,\end{equation}
leading to 
\begin{equation}\label{m0}\Lambda[x_i]=\lim_{t\rightarrow\infty} \frac{1}{t}\int_0^t r_i(\omega \tau)d\tau=\bar{r}_i,\end{equation}
where
$$\bar{r}_i=\frac{1}{2\pi}\int_0^{2\pi}r_i(\theta)d\theta,\;\;1\leq i \leq N$$
are the {\it{local}} average growth rates in each of the patches. Patch $i$ is called a {\it{source}} if $\bar{r}_i>0$ and 
a {\it{sink}} if $\bar{r}_i<0$.

The study of growth rates $\Lambda[x_i]$ when the patches are coupled 
through dispersal ($m>0$) is more difficult than in the uncoupled case, since the equations (\ref{eqsys1}) cannot be solved in closed form. A fundamental fact is that, when dispersal is present, the growth rates of all components
$\Lambda[x_i]$ are equal, and moreover 
they do not depend on the initial condition - see Section \ref{principal}. We will therefore 
denote the growth rate $\Lambda[x_1]=\Lambda[x_2]=\cdots =\Lambda[x_N]$ corresponding to the system (\ref{eqsys1}) by $\Lambda = \Lambda(m,\omega)$.

We will say that {\it{dispersal-induced growth}} (DIG) occurs 
if all patches are sinks ($\bar{r}_i<0,\;\;1\leq i\leq N$), but  $\Lambda(m,\omega)>0$.
This means that each of the populations would become extinct if isolated, but dispersal, at an appropriate rate, induces
exponential growth in all patches.

\begin{figure}
	\begin{center}
		\includegraphics[width=0.3\linewidth]{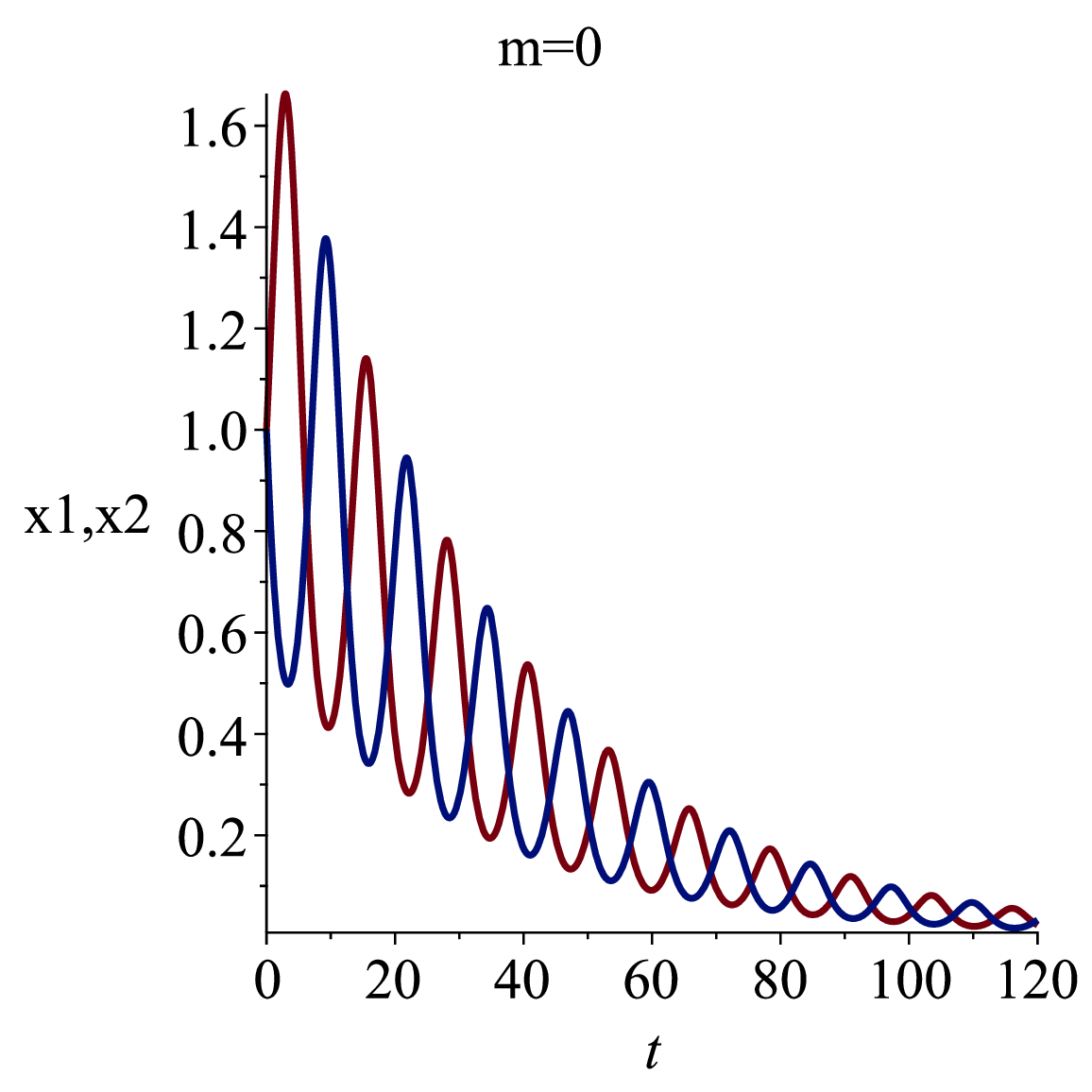}
		\includegraphics[width=0.3\linewidth]{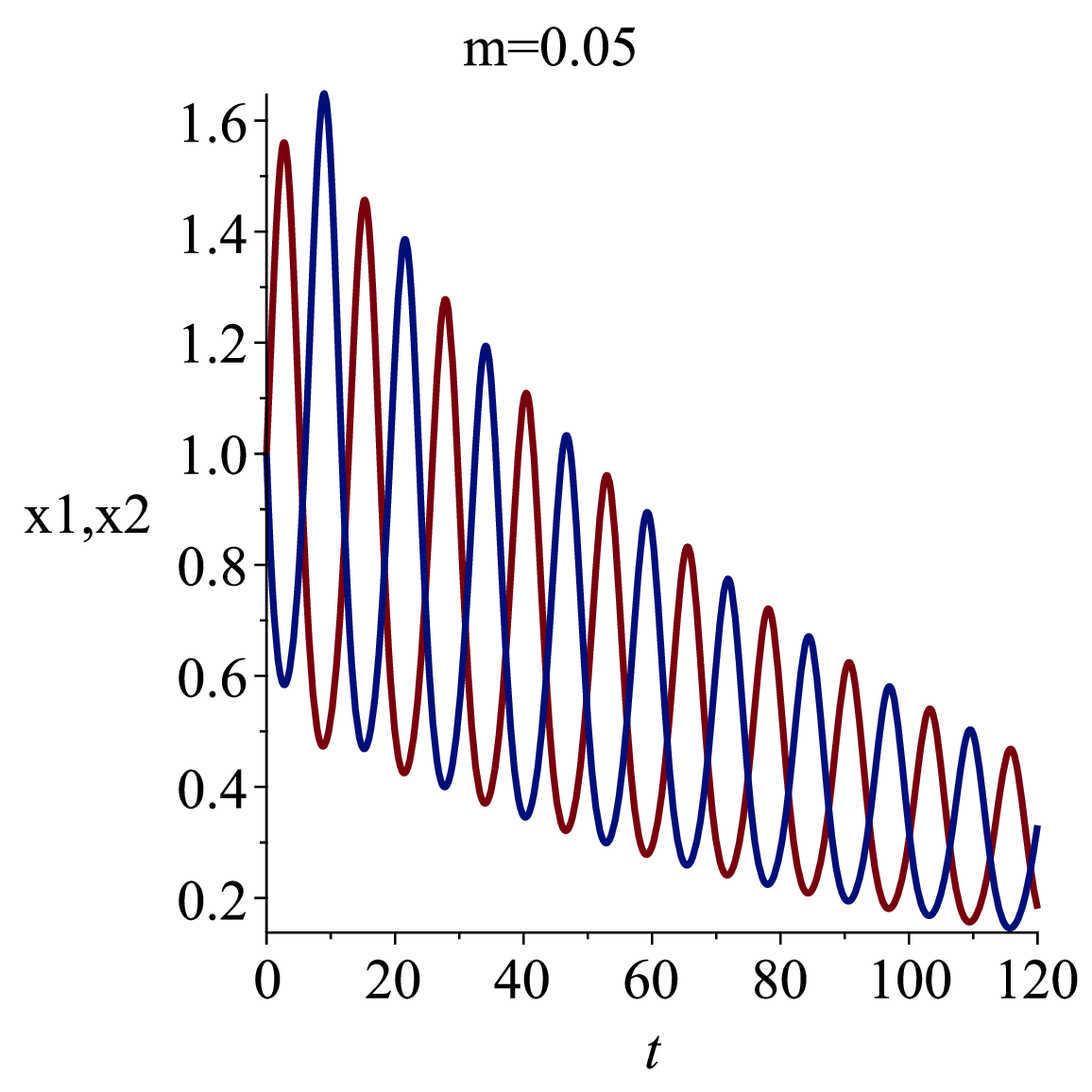}\\
		\includegraphics[width=0.3\linewidth]{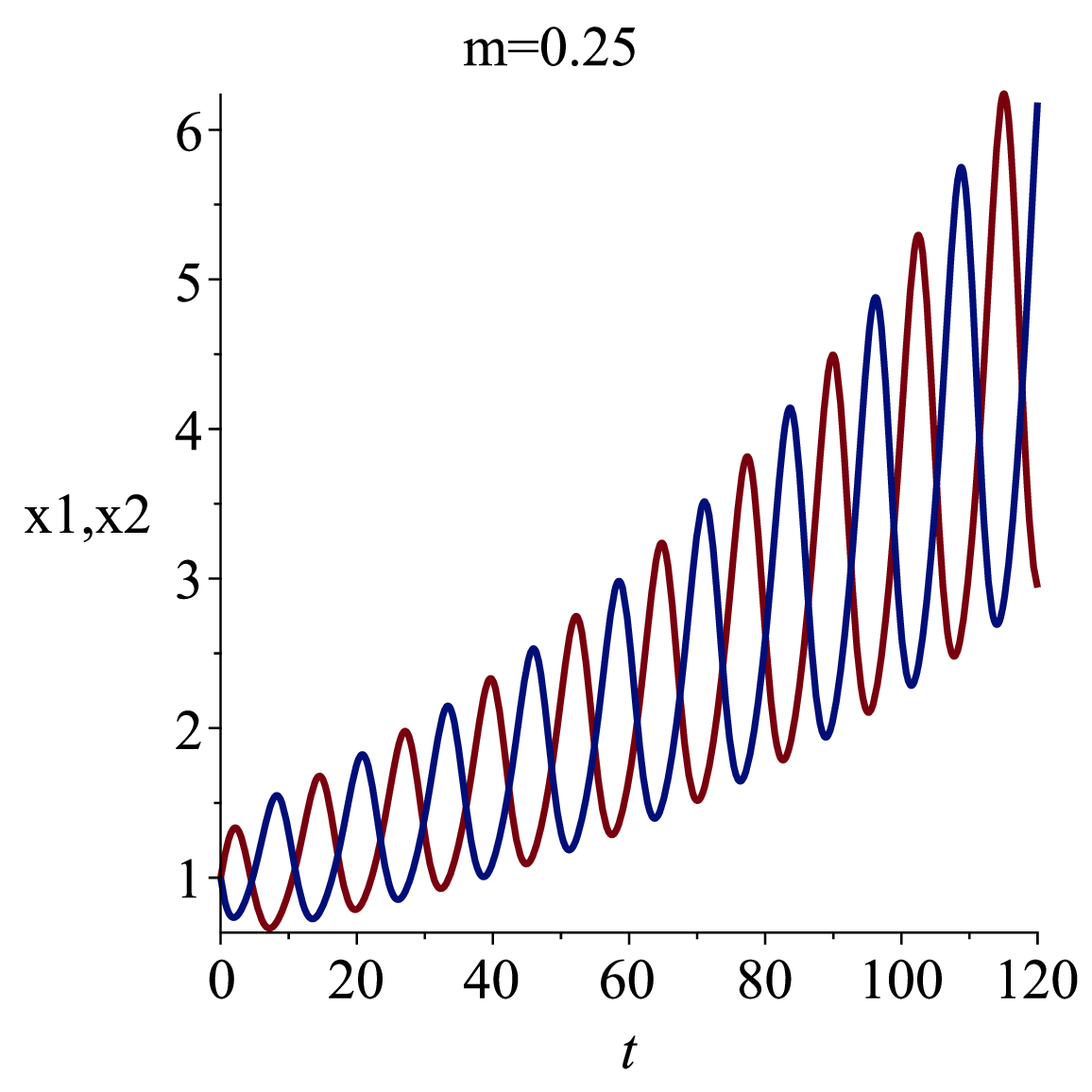}
		\includegraphics[width=0.3\linewidth]{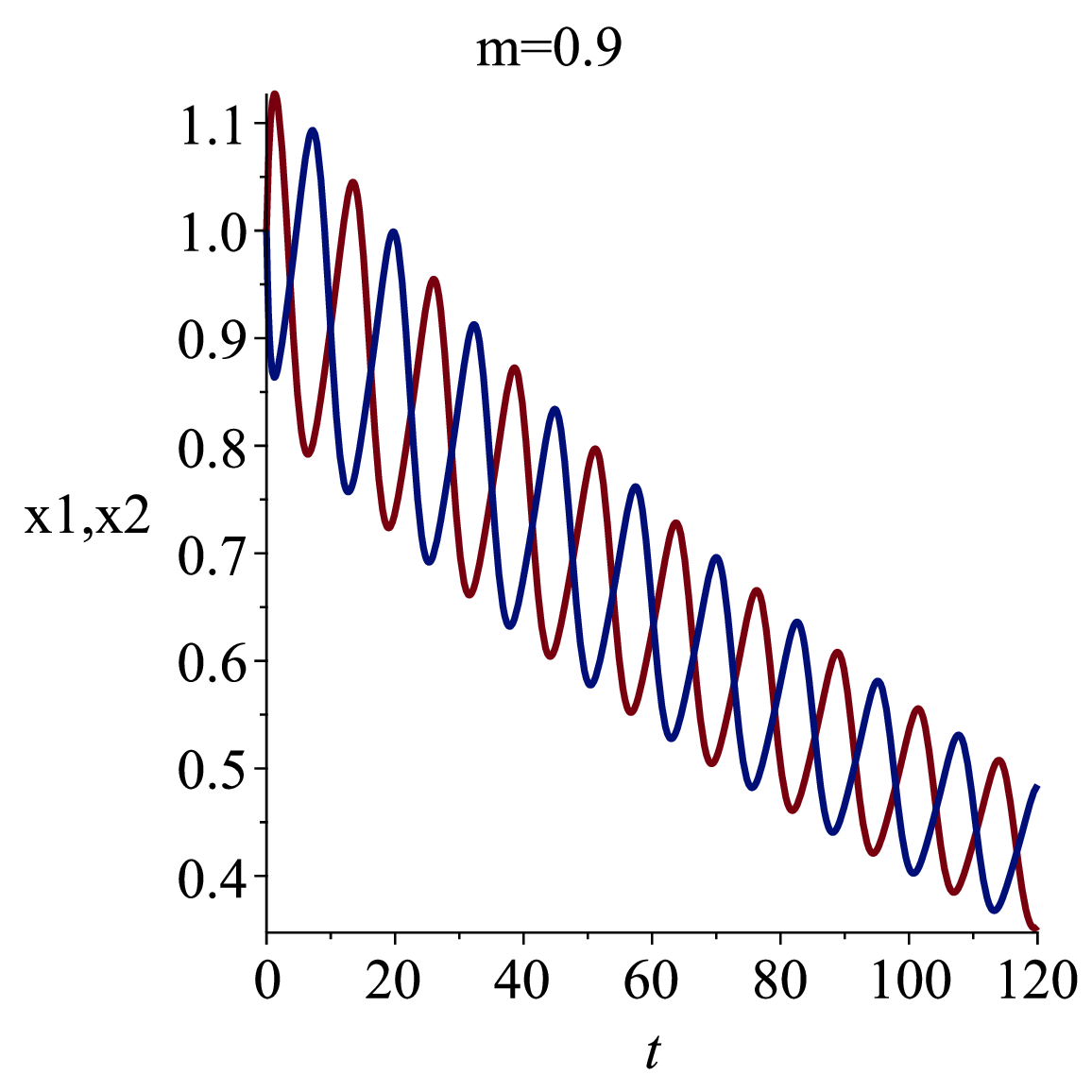}
		\caption{Solutions of (\ref{eq1}),(\ref{eq2}), with
			$r_1(\omega t)=-0.03+0.3\cos(\omega t),\;\;r_2(\omega t)=r_1(\omega t-\pi)$,
			with $\omega=0.5$, for different values of the dispersal rate $m$.}
		\label{exinf}
	\end{center}
\end{figure}

In the special case in which growth rates are constant in time ($R(\theta)\equiv R$), the growth rate $\Lambda$ is the largest eigenvalue of the 
matrix $R+mL$, and in this case we have
$\Lambda \leq  \max_{1\leq i\leq N}\bar{r}_i$ (see Lemma \ref{matrixa}(iii) below),
and in particular $\Lambda<0$ when all patches are sinks. Thus the DIG phenomenon {\it{cannot}} occur
when growth rates do not vary in time.

Studying the growth or decay
of the solutions of the system (\ref{eqsys1}) can be formulated as a question of {\it{Floquet Theory}} (\cite{chicone, hale,klausmeier}). Any system of linear differential 
equations with periodic coefficients has associated {\it{Floquet exponents}}, and indeed, in the notation used here, the value
$\Lambda$ is precisely the maximal
Floquet exponent - whose sign determines the growth or decay of the solutions (see Section \ref{principal}). However,  in contrast to time-independent 
(autonomous) systems, studying the Floquet exponents of periodic systems analytically is 
challenging, and our aim is to perform such
a study for the particular system of interest here.


\subsection{Some numerical results}
\label{numerical}

A numerical demonstration of the DIG effect is shown in Figure \ref{exinf}, in the case $N=2$, for a pair of periodic growth-rate profiles with $\bar{r}_1<0,\bar{r}_2<0$.
In the absence of dispersal ($m=0$), as well as for sufficiently weak dispersal ($m=0.05$), both populations decay, while in the presence of stronger dispersal ($m=0.25$) 
both populations grow - the DIG phenomenon. For
yet stronger dispersal ($m=0.9$) the populations once again decay.

In Figure \ref{psinksink}, we display a numerically-generated plot of the $(m,\omega)$ parameter-plane, for the same periodic growth profiles as in Figure \ref{exinf} - the computation  by which this figure was generated is explained in Section \ref{reduction}. The green region is
the set of parameter values for which $\Lambda(m,\omega)>0$, that is for which DIG occurs.
The lines are level curves of the function 
$\Lambda(m,\omega)$. 
As observed in the figure, there exists 
a value $m^*>0$ such that, if the dispersal rate satisfies  $m\in (0,m^*)$,
DIG occurs when the frequency $\omega$ of environmental variation is sufficiently small, $0<\omega <\omega_{c}(m)$, and does not occur if $\omega>\omega_{c}(m)$, or if $m\geq m^*$. 

Defining $\omega^*= \max_{m\in [0,m^*]}\omega_{c}(m)$, we see that if we fix a frequency $\omega \in (0,\omega^*)$ then DIG will occur for an intermediate range of values of $m$ - neither too weak nor too 
strong. If $\omega>\omega^*$ DIG will {\it{not}} occur 
for any dispersal rate.
The bottom part of Figure \ref{psinksink} shows the growth rate $\Lambda$ as 
a function of the dispersal rate $m$, for three values of the frequency, using the same profiles $r_1(\theta),r_2(\theta)$. Positive values of 
$\Lambda$, corresponding to DIG, occur for intermediate values of $m$.

A central aim of this work is to
obtain theoretical understanding of the features noted above, as observed in Figure \ref{psinksink}, that is to analytically derive them, thus providing a mathematical explanation for these features and proving that they are generic.

\subsection{The all-sink case}

We now present our main results for the case in which all patches are sinks ($\bar{r}_i<0$), and then discuss their
implications for characterizing the conditions under which DIG occurs.

\begin{theorem}[All-sink case] \label{sinksink} Assume
	$r_i(\theta)$ ($1\leq i\leq N$) are continuous
	$2\pi$-periodic functions, with
	$\bar{r}_i<0$. 
	Define
	$$r_{max}(\theta)=\max_{1\leq i\leq N} r_i(\theta)$$
	\begin{equation}\label{dchi}\chi = \frac{1}{2\pi}\int_0^{2\pi}r_{max}(\theta)d\theta.\end{equation}
	Then we have 
	\begin{equation}\label{limit}
		\lim_{m\rightarrow 0+}\lim_{\omega\rightarrow 0+}
		\Lambda(m,\omega)=\chi, 
	\end{equation}
	and
	
	(I) If $\chi<0$ then $\Lambda(m,\omega)<0$ (decay) for all $m>0,\omega>0$.
	
	(II) If $\chi>0$ then, defining 
	\begin{equation}\label{l0}\Lambda_0(m)=\frac{1}{2\pi}\int_0^{2\pi} \lambda(R(\theta)+mL)d\theta,\end{equation}
	where $\lambda(A)$ denotes the maximal eigenvalue of a symmetric matrix $A$, the equation
	\begin{equation}
		\label{al}
		\Lambda_0(m)=0,\;\; m>0
	\end{equation}
	has a unique solution $m=m^*>0$, and there exists a 
	continuous function $\omega_{c}:[0,m^*]\rightarrow [0,\infty)$, real-analytic on $(0,m^*)$,
	with $\omega_c(0)=\omega_c(m^*)=0$ and $\omega_c(m)>0$ for 
	$m\in (0,m^*)$, such that
	
	\begin{itemize}
		\item[a.] If $m\in (0,m^*)$ then 
		$\Lambda(m,\omega)>0$ (growth) for $\omega<\omega_c(m)$ and
		$\Lambda(m,\omega)<0$ (decay) for $\omega>\omega_c(m)$.
		
		\item[b.] If $m\geq m^*$ then $\Lambda(m,\omega)<0$
		(decay) for any $\omega>0$.  
	\end{itemize}
\end{theorem}

\begin{example}\label{ex2} In the case $N=2$, with (\ref{eq1}),(\ref{eq2}),
	it is easy to compute (\ref{l0}) explicitly and find that 
	\begin{equation}\label{l02}\Lambda_0(m)
		\doteq \frac{1}{2}\left[\bar{r}_1+\bar{r}_2+ \frac{1}{2\pi}\int_0^{2\pi}\sqrt{(r_1(\theta)-r_2(\theta))^2+4m^2} d\theta\right]-m.
	\end{equation}
	For example, taking $r_1(\theta),r_2(\theta)$ as in 
	Figure 2, we can use (\ref{l02}) and solve (\ref{al}) numerically, to find $m^*=0.7276..$ - this is the least  upper bound of values of $m$ for which DIG occurs, in agreement with the top part of Figure 2.
\end{example}

The proof of Theorem \ref{sinksink} is given in Section \ref{mainproof}, after we develop the needed tools in 
the sections \ref{growth_rate}, \ref{low_high}.
We now discuss the insights that this theorem gives into the DIG effect. This can be compared with
the observations regarding the numerical results made in Section \ref{numerical} above.
\begin{figure}
	\begin{center}
		\includegraphics[width=0.8\linewidth]{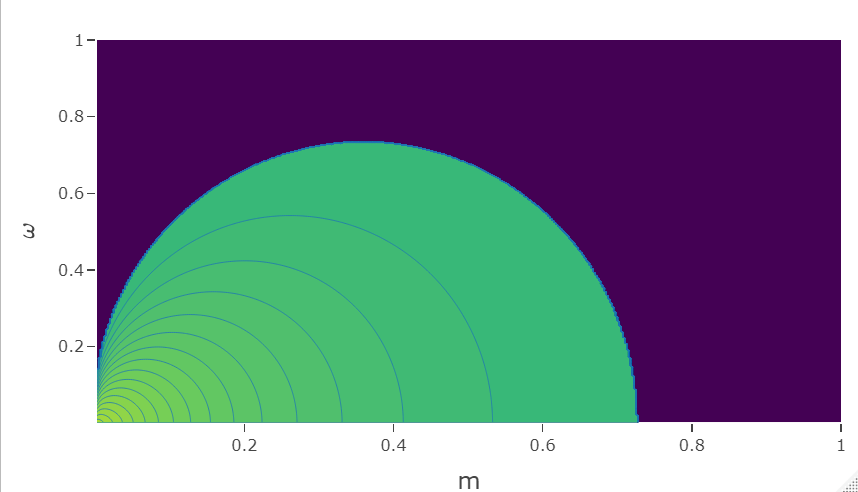}
		\includegraphics[width=0.8\linewidth]{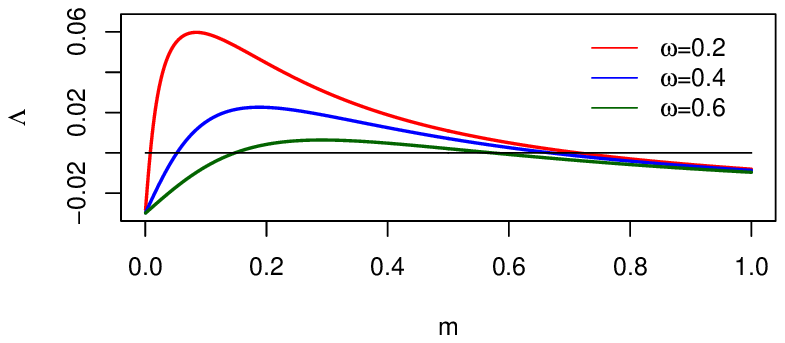}
		\caption{The two-sink case. Top: Regions of growth (green) and decay (dark) for solutions of  (\ref{eq1}),(\ref{eq2}) in the $(m,\omega)$ parameter plane. Also shown are level curves of the function $\Lambda(m,\omega)$.
			Bottom:	Growth rate $\Lambda$ of solutions of  (\ref{eq1}),(\ref{eq2}) as a function of the dispersal rate $m$ for three values of the frequency $\omega$.
			Here $r_1(\omega t)=-0.03+0.3\cos(\omega t),\;\;r_2(\omega t)=r_1(\omega t-\pi)=-0.03-0.3\cos(\omega t)$.
		}\label{psinksink}
	\end{center}
\end{figure}

\vspace{0.5cm}

(1) {\bf{The periodic growth profiles}}: 
The condition $\chi>0$, where $\chi$ is given by
(\ref{dchi}), is {\it{necessary}} for DIG to occur.
$\chi$ is the time-average of 
the maximum of the local growth rates at each point in time.
This condition entails in particular that:

(i) At least one of the local growth rates must be positive at {\it{some}} times. However it is possible for all but one of the local growth 
rates to be negative at {\it{all}} times, and indeed 
it is even possible for all but one of the local growth rates to be 
negative {\it{and}} constant (time-independent) - see Figure \ref{const} (left) for such an example.

(ii) None of the local growth rates is higher than all others at
all times. Indeed if, {\it{e.g.}}, $r_1(\theta)\geq r_i(\theta)$ for all $i$ and all $\theta$, then we have $r_{max}(\theta)=r_1(\theta)$, so $\chi\leq \bar{r}_1<0$.

(iii) The local growth rates $r_i(\theta)$ cannot be
too similar (in particular they cannot be identical). 
Indeed, if $\max_{\theta}| r_i(\theta)-r_j(\theta)| \leq \epsilon$  ($1\leq i,j\leq N$) then, for all $i$,
$$\chi=\frac{1}{2\pi}\int_0^{2\pi}r_{max}(\theta)d\theta\leq 
\frac{1}{2\pi}\int_0^{2\pi}(r_i(\theta)+\epsilon)d\theta=\bar{r}_i+\epsilon.$$
Thus, if $\epsilon<\max_{1\leq i\leq N} |\bar{r}_i|$ then
$\chi<0$. We therefore conclude that if
$$\max_{1\leq i,j\leq N}\max_{\theta}| r_i(\theta)-r_j(\theta)| <\max_{1\leq i\leq N}|\bar{r}_i|$$
then DIG {\it{cannot}} occur. Note, however, that it is possible to have
$\chi>0$ even when all $r_i(\theta)$ are phase-synchronized, so 
that DIG can occur even if the same seasonal effect 
acts in all patches, as long as the strength of this effect is not identical in all patches - see Figure \ref{const} (right) for an example.

\begin{figure}
	\begin{center}
		\includegraphics[width=0.45\linewidth]{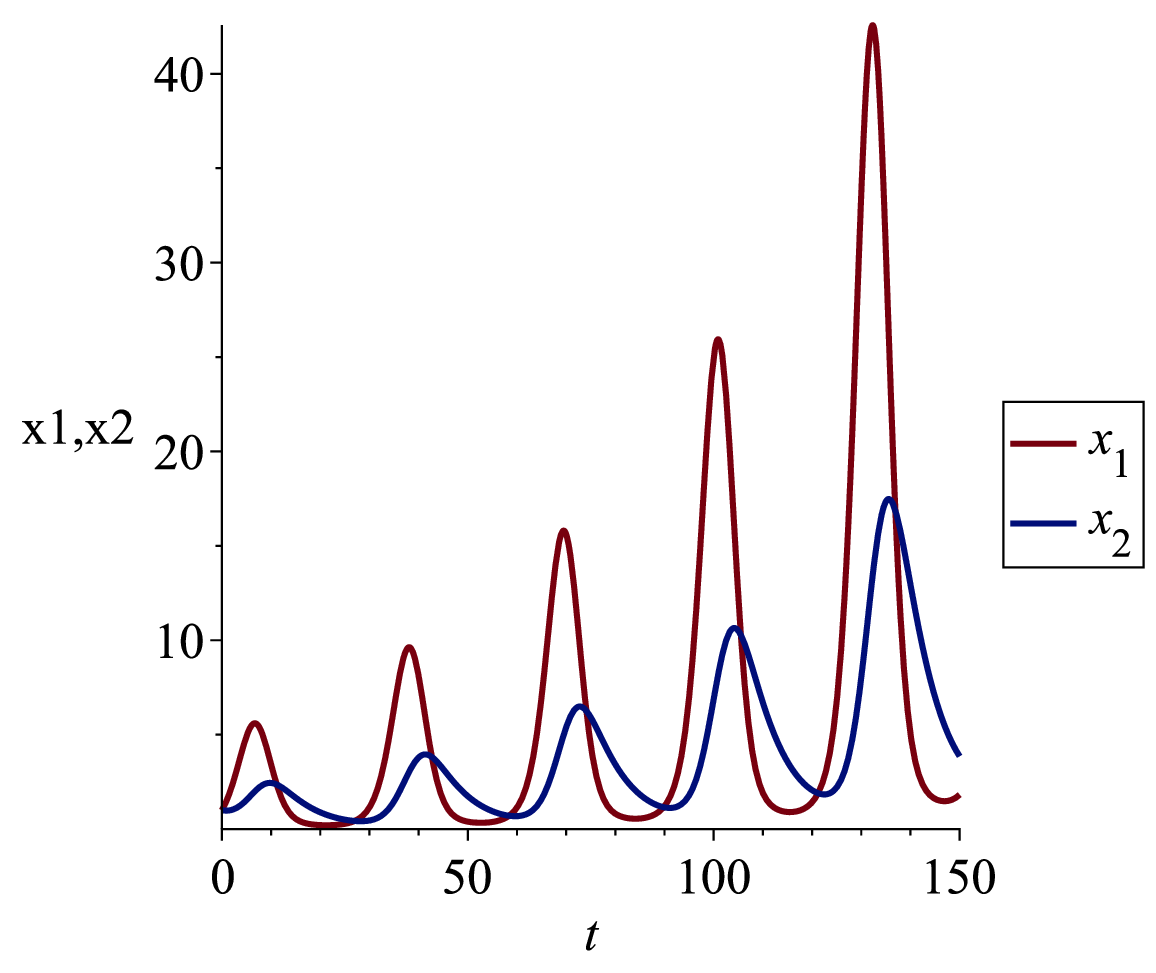}
		\includegraphics[width=0.45\linewidth]{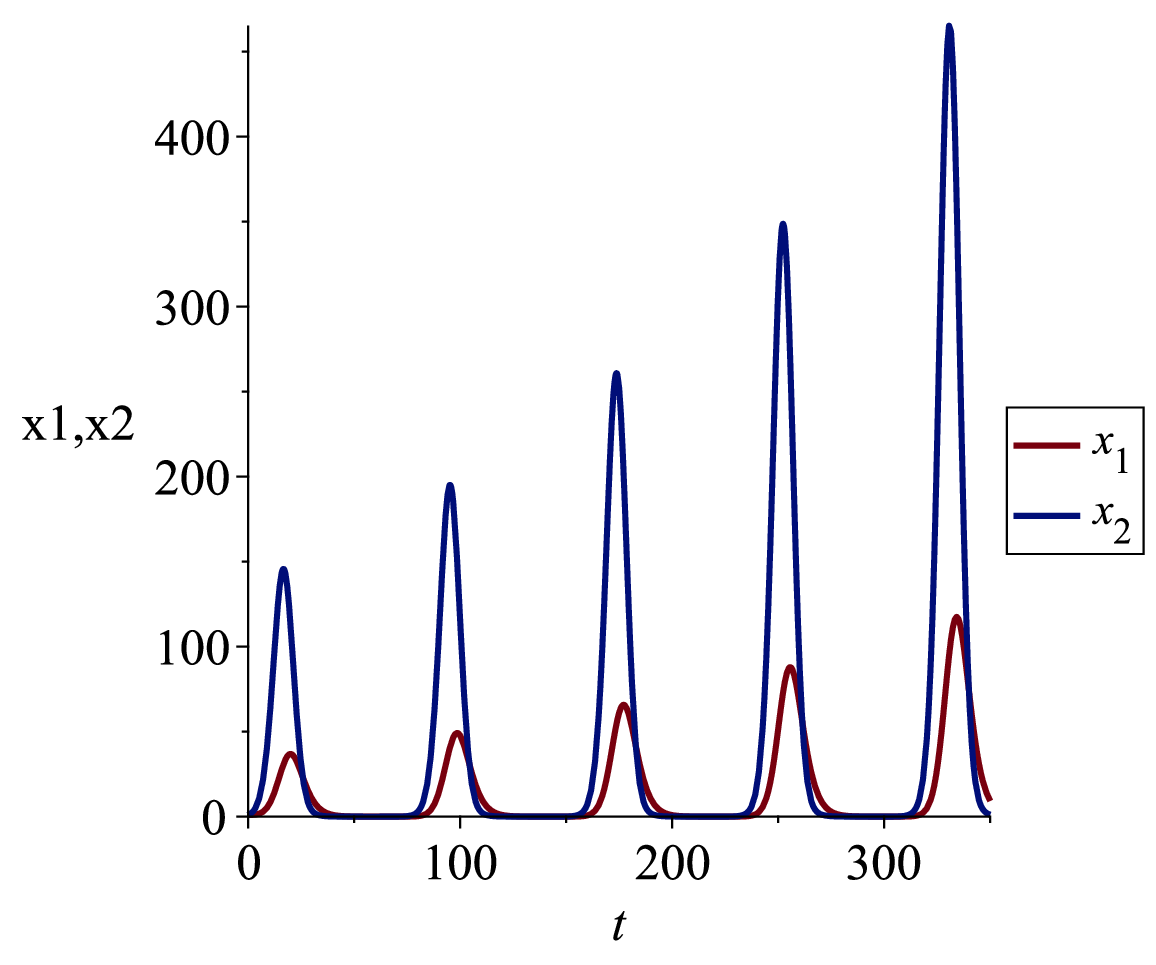}
		\caption{Left: Example of the DIG effect when growth rate in one of two 
			patches is constant: $r_1(\omega t)=-0.05,\;\;r_2(\omega t)=-0.05+0.5\cos(\omega t)$,
			$\omega=0.2,m=0.1$.
			Right: Example of DIG when the seasonal effect in two patches is phase-synchronized:  $r_1(\omega t)=-0.1+0.1\cos(\omega t),\;\;r_2(\omega t)=-0.1+0.6\cos(\omega t)$,
			$\omega=0.08,m=0.05$.
		}\label{const}
	\end{center}
\end{figure}

(2) {\bf{Rate of dispersal}}. Assuming that 
$\chi>0$, part (II) of Theorem \ref{sinksink} 
implies that when $m\in (0,m^*)$ DIG will occur for $\omega<\omega_c(m)$.
Thus a dispersal rate which is too large 
($m> m^*$) will prevent the DIG effect from occurring, but dispersal rates
$m<m^*$ will induce DIG {\it{provided}}
that the frequency $\omega$ is sufficiently small.
The qualification in the previous sentence is essential: if we {\it{fix}} $\omega>0$, then, since $\omega_c(0)=0$,
we will have $\omega_c(m)<\omega$ for $m$ sufficiently small, implying that $\Lambda(m,\omega)<0$, so that DIG does not occur.

(3) {\bf{Frequency of oscillations}}: 
If $\omega>\omega^*=\max_{m\in [0,m^*]}\omega_c(m)$, 
part (II) of Theorem \ref{sinksink} implies that
DIG will not occur. Thus, while time-variation of at least one growth rate is essential for DIG, the frequency of this variation cannot be too high.

\subsection{The Source-Sink case}

\begin{figure}
	\begin{center}
		\includegraphics[width=0.8\linewidth]{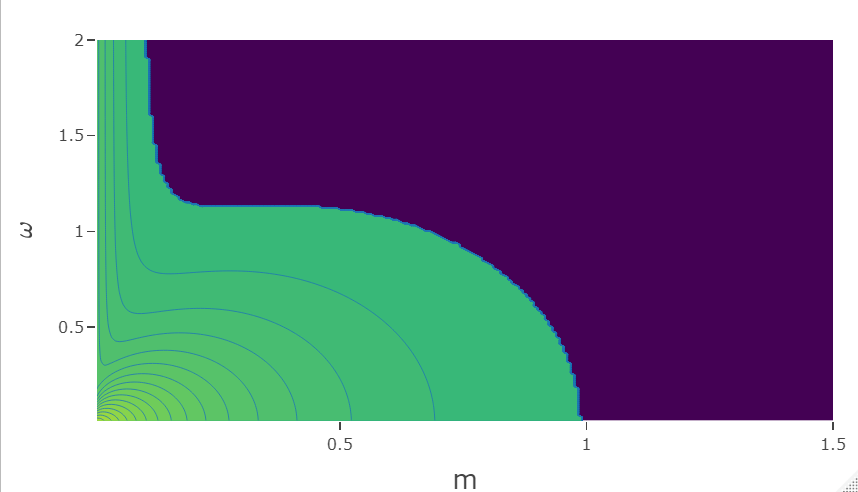}
		\caption{The Source-Sink case: Regions of growth (green) and decay (dark) for solutions of  (\ref{eq1}),(\ref{eq2}) in the $(m,\omega)$ parameter plane. Also shown are level curves of the function $\Lambda(m,\omega)$. Here $r_1(\omega t)=-0.1+0.3\cos(\omega t),\;\;r_2(\omega t)=0.05-0.3\cos(\omega t)$.
		}\label{psourcesink}
	\end{center}
	
\end{figure}

Although our main interest is in the DIG effect, which involves the all-sink 
case, we complement our analysis with 
a treatment of the source-sink case, 
using the same methods.
Our results in this case are given by

\begin{theorem}[Source-Sink case]\label{sourcesink} Assume
	$r_i(\theta)$ ($1\leq i\leq N$)  are continuous
	$2\pi$-periodic functions, with
	\begin{equation}\label{pn}\min_{1\leq i\leq N}\bar{r}_i<0,\;\;\;\max_{1\leq i\leq N}\bar{r}_i>0.\end{equation}
	Denote
	\begin{equation}\label{ra}\bar{R}=\left(\begin{array}{ccc}
			\bar{r}_1	& 0 & 0 \\
			0	& \ddots &  0\\
			0	& 0 & \bar{r}_N
		\end{array} \right),\;\;\;\;\;\bar{r}=\frac{1}{N}\sum_{i=1}^n \bar{r}_i.\end{equation}
	
	(I) If $\bar{r}> 0$ then for any $m>0, \omega>0$ we have $\Lambda(m,\omega)>0$ (growth).

	(II) If $\bar{r}<0$ then the equation (\ref{al})
	has a unique solution $m=m^*>0$. 
	
	Defining
	\begin{equation}\label{domi}\Lambda_\infty(m)=\lambda\left(\bar{R}+mL\right),
	\end{equation}
	where $\lambda(A)$ denotes the maximal eigenvalue of a symmetric matrix $A$, the equation
	\begin{equation}\label{li0}
		\Lambda_\infty(m)=0,\;\;\;m>0
	\end{equation}
	has a unique solution, which we denote by $\hat{m}$, and
	we have $\hat{m}\leq m^*$. Unless  $r_i(\theta)-r_j(\theta)$ are constant for all $i,j$, we have strict inequality $\hat{m}<m^*$, and
	there exists a 
	continuous function $\omega_{c}:(\hat{m},m^*]\rightarrow [0,\infty)$, real-analytic on $(\hat{m},m^*)$,
	with $\lim_{m\rightarrow \hat{m}+}\omega_c(m)=+\infty$, $\omega_c(m^*)=0$,
	and $\omega(m)>0$ for $m\in(\hat{m},m^*)$,
	such that
	\begin{itemize}
		\item[a.] If $m\in (0,\hat{m}]$ then for any $\omega>0$ we have $\Lambda(m,\omega)>0$ (growth).
		
		\item[b.] If $m\in (\hat{m}, m^*)$ then for $\omega \in (0,\omega_c(m))$ we have $\Lambda(m,\omega)>0$ (growth) and 
		for $\omega>\omega_c(m)$ we have $\Lambda(m,\omega)<0$
		(decay).
		
		\item[c.] If $m\geq m^*$ then for all $\omega>0$ we have $\Lambda(m,\omega)<0$
		(decay).
	\end{itemize}
	
\end{theorem}

The proof of Theorem \ref{sourcesink} is given in Section \ref{mainproof}.

Part (I) of the above theorem says that when the mean of the 
time-averaged local growth rates is positive, growth always occurs.
Part (II) says that growth may occur also when the mean of 
the time-averaged growth rates is negative, and -  
in contrast with the all-sink case - here, if $m$ is sufficiently
small ($m<\hat{m}$), we have $\Lambda(m,\omega)>0$ (growth) for {\it{all}} $\omega>0$. This can be seen 
in the parameter-plane diagram in Figure \ref{psourcesink}, obtained numerically. 
Let us note that the condition $m<\hat{m}$ 
is precisely the condition for growth
in the case in which the local growth rates $r_i(\theta)$ are {\it{constant}}, with 
values $\bar{r}_i$ (with $\min_{1\leq i\leq N}\bar{r}_i<0$, 
$\max_{1\leq i\leq N}\bar{r}_i>0$), that is the condition under 
which the largest eigenvalue of the matrix
$\bar{R}+mL$ is positive. However, unless all
$r_i(\theta)-r_j(\theta)$ are constant, we have $\hat{m}<m^*$, and the theorem
implies that, when the frequency is sufficiently small, the time periodic system also displays growth for parameter 
values $m\in (\hat{m},m^*)$ for which the corresponding time-averaged system leads to decay.

We note that the special case in which all $r_i(\theta)-r_j(\theta)$ are constant, which was excluded in the above theorem, is a trivial one: in this case we have
$R(\theta)=\bar{R}+r_0(\theta)I$, where $r_0(\theta)$ 
is $2\pi$-periodic and satisfies $\int_0^{2\pi}r_0(\theta)d\theta=0$, and then (\ref{eqsys1}) has the explicit solution $\bold{x}(t)=e^{\int_0^t r_0(\omega s)ds}e^{t(\bar{R}+mL)}\bold{x}(0)$, from which it follows that $\Lambda(m,\omega)=\lambda(\bar{R}+mL)=\Lambda_\infty(m)$ for all $\omega$, so that the growth rate is identical to that of the corresponding time-averaged system for all $(m,\omega)$.

\begin{example}\label{ex3}
	In the case $N=2$, with (\ref{eq1}),(\ref{eq2}),
	and assuming the source-sink case $\bar{r}_1\bar{r}_2<0$, it is easy to compute (\ref{domi}) explicitly and find that 
	$$\Lambda_\infty(m)=\frac{1}{2}\left[\bar{r}_1+\bar{r}_2+\sqrt{(\bar{r}_1-\bar{r}_2)^2+4m^2} \right] -m.
	$$
	The solution $\hat{m}$ of (\ref{li0}) is 
	then given by
	\begin{equation}\label{m2}\hat{m}=\left(\frac{1}{\bar{r}_1}+\frac{1}{\bar{r}_2} \right)^{-1}.\end{equation}
	For example, taking $r_1(\theta),r_2(\theta)$ as in Figure 4, (\ref{m2}) gives $\hat{m}=0.1$,
	so that by Theorem \ref{sourcesink}, for $m\leq{\hat{m}}$ we have $\Lambda(m,\omega)>0$
	for all $\omega$.
	Using (\ref{l02}) and solving (\ref{al}) numerically, we find $m^*=0.990..$. Thus for
	$\hat{m}<m<m^*$ we have   $\Lambda(m,\omega)>0$
	for $\omega$ small, and  $\Lambda(m,\omega)<0$
	for $\omega$ large, and for $m>m^*$ we always have $\Lambda(m,\omega)<0$. All these results are in agreement with Figure 4.
\end{example}

\section{Characterizations of the growth rate}
\label{growth_rate}

In this section we characterize the growth rate $\Lambda(m,\omega)$  corresponding to (\ref{eqsys1}) in different ways, each of which has its uses: as the dominant eigenvalue of 
a monodromy matrix, as the principal eigenvalue of a periodic problem, and, in the case of two patches, via an integral related to a periodic solution of an associated  nonlinear scalar differential equation.  

\subsection{The growth rate as a principal eigenvalue}
\label{principal}

The fundamental solution corresponding to (\ref{eqsys1}) is 
the matrix function $X(t)$ satisfying
$$X'(t)=[R(\omega t)+mL]X(t),\;\;\;X(0)=I,$$
where $I$ is the identity matrix.
Since the matrix $R(\omega t)+mL$ is an irreducible cooperative matrix (i.e. has non-negative non-diagonal entries), the Kamke-M\"uller theorem \cite{hirsch} implies that the matrices $X(t)$ have positive entries - so that 
solutions of (\ref{eqsys1}) with positive initial conditions remain positive for all time.
The Perron-Frobenius 
theorem implies that the 
matrix $X(T)$ ($T=\frac{2\pi}{\omega}$), known as the monodromy
matrix, has a dominant eigenvalue (an eigenvalue of maximal modulus) $\rho=\rho(X(T))$ which is positive and simple, and the corresponding eigenvector $\bold{v}$ has positive entries and is the only positive eigenvector of $X(T)$. 
By Floquet's Theorem (see, {\it{e.g.}}, \cite{chicone}, Theorem 2.83)
we have $X(t+T)=X(t)X(T)$ for all $t$.
Therefore, defining 
\begin{equation}\label{lambda}\lambda=\frac{1}{T}\ln \left( \rho(X(T))\right),\;\;\varphi(t)=e^{-\lambda t}X(t)\bold{v},\end{equation}
($\lambda$ is known as the Floquet exponent),  we have that 
$$\varphi\left(t+T\right)=e^{-\lambda \left(t+T \right)}X\left(t+T \right)\bold{v}=\frac{1}{\rho}e^{-\lambda t}X(t)X(T)\bold{v}=e^{-\lambda t}X(t)\bold{v}=\varphi(t),$$
$$\varphi'(t)=-\lambda e^{-\lambda t}X(t)\bold{v}+e^{-\lambda t}X'(t)\bold{v}=[R(\omega t)+mL]\varphi(t)-\lambda \varphi(t).
$$
Thus $\lambda$ given by (\ref{lambda}) is the 
principal eigenvalue (the one with largest real part) of the periodic problem
$$\varphi'(t)=[R(\omega t)+mL]\varphi(t)-\lambda \varphi(t),\;\;\;\;\varphi\left(t+T\right)=\varphi(t).
$$
Moreover, we have that, for any positive solution $\bold{x}(t)$ of (\ref{eqsys1}):
\begin{equation}\label{pert}\underline{C}\leq \frac{x_i(t)}{e^{\lambda t}\varphi_i(t)}\leq \overline{C},\;\;\;\;\;1\leq i\leq N,\end{equation}
where the positive constants 
$\underline{C},\overline{C}$ depend on the initial conditions - for an elegant proof of this fact using the relative entropy method see
\cite{perthame}, Sec. 6.3.2. (\ref{pert}), 
together with the periodicity of $\varphi(t)$ imply
\begin{equation}\label{flo}\Lambda[x_i]=\lim_{t\rightarrow\infty}\frac{1}{t}\ln(x_i(t))=\lambda.\end{equation}
In particular this shows that, under coupling, the growth rates in all patches are identical. We note also that by the analytic dependence of solutions of 
differential equations on parameters, and the fact that
$\rho(X(T))$ is a simple eigenvalue, (\ref{flo}) implies that the function $\Lambda(m,\omega)$ is real-analytic.

We can also normalize the period, setting 
$\varphi(t)=\bold{u}(\omega t)$, where $\bold{u}(\theta)$ is
$2\pi$-periodic, and we thus obtain
\begin{lemma}\label{connect}
	The growth rate $\Lambda(m,\omega)$ is given as 
	$$\Lambda(m,\omega)=\lambda,$$
	where $\lambda$ is the principal eigenvalue of the periodic problem
	\begin{eqnarray}\label{ep1}&&\omega \bold{u}'(\theta)=[R( \theta)+mL]\bold{u}(\theta)-\lambda \bold{u}(\theta),\nonumber\\
		&&\bold{u}(\theta+2\pi)=\bold{u}(\theta),
	\end{eqnarray}
	that is the eigenvalue with largest real part.
\end{lemma}

We now cite an important result from \cite{liu} 
which will play a significant role in the proofs of the main results. 

\begin{lemma}\label{monotone} For all $m>0,\omega>0$ we have $\Lambda_\omega'(m,\omega)\leq 0$. The inequality is strict
	except in the case that $r_i(\theta)-r_j(\theta)$ are constants for all $i,j$, so that
	$\Lambda(m,\omega)$ is strictly decreasing in $\omega$ for any fixed $m$.
\end{lemma}

In view of Lemma \ref{connect}, this result follows from
Part (ii) of Theorem 1.1 in \cite{liu} and Remark 1.2 following that theorem. Note that in the formulation of the results in \cite{liu} the principal eigenvalue is defined as the negative of the value  as defined here, and is thus increasing with respect to $\omega$.

\subsection{An associated scalar differential equation in the two-patch case}
\label{reduction}

In the case of two patches ($N=2$) we now obtain a formula for the growth rate $\Lambda$ in terms of 
the periodic solution of an associated nonlinear scalar differential equation. This formula, besides its intrinsic interest, has been useful for us in carrying out numerical computations.

\begin{lemma}
	If $(x_1(t),x_2(t))$  is any positive solution of (\ref{eq1}),(\ref{eq2}),
	then the function
	\begin{equation}\label{defz}z(t)=\frac{x_2(t)}{x_1(t)}.\end{equation}
	satisfies the differential equation
	\begin{equation}\label{dz}
		z'=(r_2(\omega t)-r_1(\omega t))z +m( 1-z^2).
	\end{equation}
\end{lemma}

\begin{proof}
	Dividing (\ref{eq1}),(\ref{eq2}) by $x_1$ we have
	\begin{eqnarray}\label{d2}\frac{x_1'(t)}{x_1(t)}&=& r_1(\omega t) +m z(t)-m,\\
		\frac{x_2'(t)}{x_1(t)}&=& r_2(\omega t)z(t) +m -m\nonumber z(t),\end{eqnarray}
	from which it follows that
	\begin{eqnarray*}z'(t)&=&\frac{x_1(t) x_2'(t)-x_1'(t) x_2(t)}{x_1(t)^2}=\frac{x_2'(t)}{x_1(t)}-\frac{x_1'(t) }{x_1(t)}\cdot z(t)\\
		&=&r_2(\omega t)z(t) +m -m z(t)-(r_1(\omega t) +m z(t)-m)\cdot z(t)\\&=&(r_2(\omega t)-r_1(\omega t))z(t) +m( 1-z(t)^2).\end{eqnarray*}
\end{proof}

Regarding the differential equation (\ref{dz}), we note that $z=0$ implies $z'>0$, so 
that any solution with positive initial condition
remains positive for all $t>0$.
Moreover the following lemma shows that all solutions of (\ref{dz}) 
approach a unique periodic solution as $t\rightarrow \infty$, and that this periodic solution can be used to compute the growth rate 
$\Lambda(m,\omega)$.
\begin{lemma}\label{per}
	Assume $m>0$.
	For each $\omega>0$, there exists a unique $\frac{2\pi}{\omega}$-periodic solution of 
	(\ref{dz}), which we denote by $z_p (t)$, and this solution is globally stable, that is, all solutions $z(t)$ of (\ref{dz}) with $z(0)>0$ satisfy
	\begin{equation}\label{ll}\lim_{t\rightarrow\infty}[z(t)-z_p(t)]=0.\end{equation}
	
	The growth rate of solutions of (\ref{eq1}),(\ref{eq2}) is given by
	\begin{equation}\label{kf3}\Lambda(m,\omega)=\frac{1}{2}\left(\bar{r}_1+\bar{r}_2\right)+m\cdot \left( \frac{\omega}{4\pi}\int_0^{\frac{2\pi}{\omega}}\left(z_p(\tau)+\frac{1}{z_p(\tau)}\right)d\tau -1 \right).\end{equation}
\end{lemma}
\begin{proof}
	We first transform (\ref{dz}) by 
	making the change of variable $z(t)=e^{y(t)}$ (using the fact that $z(t)$ is positive), to obtain the equation
	\begin{equation}\label{yy}y'=(r_2(\omega t)-r_1(\omega t))+m(e^{-y}-e^{y}).\end{equation}
	We will show that (\ref{yy}) has a unique $\frac{2\pi}{\omega}$-periodic solution $y_p(t)$
	which is globally stable in the sense that, for any solution $y(t)$
	of (\ref{yy}),
	$$
	\lim_{t\rightarrow\infty}[y(t)-y_p(t)]=0.
	$$
	This will imply the result of Lemma \ref{per} for the equation (\ref{dz}), with $z_p(t)=e^{y_p(t)}$.
	Denote by $y(t,y_0)$ the solution of (\ref{yy}) satisfying the
	initial condition $y(0)=y_0$.
	We define $\psi$ to be the time-$T$ ($T=\frac{2\pi}{\omega}$) Poincar\'e mapping
	corresponding to (\ref{yy}): 
	$$\psi(y_0)=y\left(T,y_0 \right).$$
	To compute the derivative of this function we differentiate the equations
	$$y_t(t,y_0)=(r_2(\omega t)-r_1(\omega t))+m[e^{-y(t,y_0)}-e^{y(t,y_0)}],\;\;\; y(0,y_0)=y_0$$
	with respect to $y_0$, obtaining
	$$y_{y_0 t}(t,y_0)=-m[e^{-y(t,y_0)}+e^{y(t,y_0)}]y_{y_0}(t,y_0)\;\;\; y_{y_0}(0,y_0)=1,$$
	leading to
	$$y_{y_0}(t,y_0)=e^{-m\int_0^t (e^{-y(\tau,y_0)}+e^{y(\tau,y_0)})d\tau},$$
	so that 
	$$\psi'(y_0)=y_{y_0}(T,y_0)=e^{-m\int_0^T (e^{-y(\tau,y_0)}+e^{y(\tau,y_0)})d\tau}$$
	satisfies $0<\psi'(y_0)<e^{-2mT}<1$ for all $y_0$. This implies that 
	$\psi:\Real\rightarrow \Real$ is a contraction mapping, so 
	that by Banach's contraction mapping principle (\cite{hale}, Section 0.3) it has a unique fixed point $y^*\in \Real$, and, for all $y_0\in \Real$, the iterates of $\psi$ satisfy
	$\lim_{k\rightarrow \infty} \psi^k(y_0)=y^*$.
	This, in turn, implies that the function $y_p(t)=y(t,y^*)$ is 
	a $T$-periodic solution of (\ref{yy}), which is globally stable. 
	
	We now show that the periodic solution $z_p$ determines 
	the growth rate $\Lambda$ of $x_1,x_2$.
	By (\ref{d2}) and (\ref{ll}) we have
	$$\lim_{t\rightarrow \infty}\frac{1}{t}\ln\left(x_1(t) \right)=\lim_{t\rightarrow \infty}\frac{1}{t}\int_0^t \frac{x_1'(\tau)}{x_1(\tau)}d\tau$$
	$$=\lim_{t\rightarrow \infty}\frac{1}{t}\left[\int_0^t r_1(\omega \tau)d\tau +m \int_0^t(z(\tau)-1)d\tau \right]=\bar{r}_1+m\cdot \left( \frac{\omega}{2\pi}\int_0^{\frac{2\pi}{\omega}}z_p(\tau)d\tau -1 \right). $$
	Note also that, by (\ref{defz}), and since (\ref{ll}) implies that any solution $z(t)$ of (\ref{dz}) is bounded on $[0,\infty)$, we have
	$$\lim_{t\rightarrow \infty}\frac{1}{t}\ln(x_2(t))=\lim_{t\rightarrow \infty}\frac{1}{t}\left[\ln(z(t))+\ln(x_1(t))\right] =\lim_{t\rightarrow \infty}\frac{1}{t}\ln(x_1(t)).$$
	We thus have
	\begin{equation}\label{kf}\Lambda[x_1]=\Lambda[x_2]=\Lambda(m,\omega)=\bar{r}_1+m\cdot \left( \frac{\omega}{2\pi}\int_0^{\frac{2\pi}{\omega}}z_p(\tau)d\tau -1 \right).\end{equation}
	By exchanging the roles of $r_1(\theta)$ and $r_2(\theta)$,
	so that $z_p$ is replaced by $\frac{1}{z_p}$ we get the equivalent expression
	\begin{equation}\label{kf2}\Lambda(m,\omega)=\bar{r}_2+m\cdot \left( \frac{\omega}{2\pi}\int_0^{\frac{2\pi}{\omega}}\frac{1}{z_p(\tau)}d\tau -1 \right),\end{equation}
	and by averaging (\ref{kf}),(\ref{kf2}) we obtain the symmetric expression (\ref{kf3}).
\end{proof}
We note that the plots in Figures \ref{psinksink},\ref{psourcesink} were obtained by 
using Lemma \ref{per}: For each point in a 
grid in the plotted region, the periodic solution $z_p$
was found numerically, and the quantity $\Lambda(m,\omega)$
was computed using (\ref{kf}). The green (growth) region is the set of points 
for which $\Lambda(m,\omega)>0$.

\section{The low and high frequency limits of the growth rate}
\label{low_high}

The characterization of the growth rate $\Lambda$ in
terms of the principal eigenvalue of a periodic problem
(Lemma \ref{connect}) allows us to employ the recent results of \cite{liu}, which determine the limit of 
$\Lambda(m,\omega)$ in the cases $\omega\rightarrow 0$
and $\omega\rightarrow \infty$, and thus provide a key element in the proofs of Theorems \ref{sinksink},\ref{sourcesink}. In this section, we cite results from \cite{liu} and apply them to our specific problem.

\subsection{The low frequency limit}
\label{low}

We now present an explicit expression for the growth rate
$\Lambda(m,\omega)$ in the limit $\omega\rightarrow 0$ (see Figure \ref{quasistatic2} for a numerical illustration of the contents of this result).

\begin{figure}
	\begin{center}
		\includegraphics[width=0.8\linewidth]{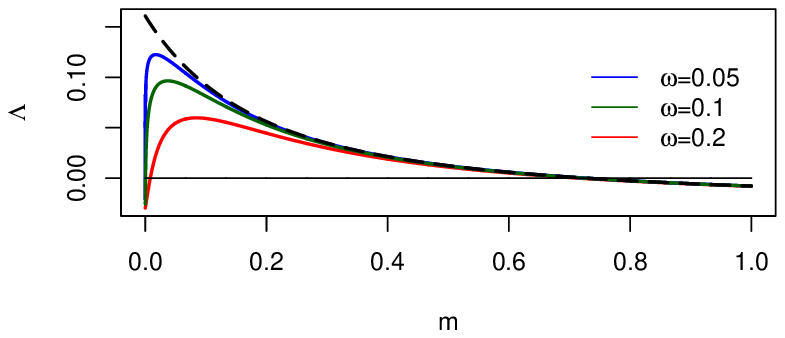}
		\caption{Illustration of the result of Lemma 5. Here
			$r_1(\omega t)=-0.03+0.3\cos(\omega t),\;\;r_2(\omega t)=r_1(\omega t-\pi)$.
			The curves $m\rightarrow \Lambda(m,\omega)$
			are plotted for $\omega=0.2,0.1,0.05$, and 
			they can be seen to converge to 
			the curve $\Lambda_0(m)$ (dashed line).
		}\label{quasistatic2}
	\end{center}
\end{figure}

Theorem 2.1 of \cite{liu} (in the case $a=0$ of that theorem), applied to the periodic eigenvalue problem (\ref{ep1}), gives:

\begin{lemma}\label{omega0}
	Let $\Lambda_0(m)$ be defined by (\ref{l0}).
	Then we have, for each $m>0$:
	$$\lim_{\omega\rightarrow 0+}\Lambda(m,\omega)=\Lambda_0(m).$$
\end{lemma}

The auxilliary results in the following lemma will be needed below.

\begin{lemma}\label{matrixa}
	Let $L$ be an $N\times N$ symmetric matrix with non-negative non-diagonal elements, which is irreducible, and satisfies (\ref{diagonal}). Let 
	$$D=\left(\begin{array}{ccc}
		d_1	& 0 & 0 \\
		0	& \ddots &  0\\
		0	& 0 & d_N
	\end{array} \right)$$
	be a diagonal matrix. Then
	
	(i) If the numbers $d_i$ ($1\leq i\leq N$) are not all equal to each other, then $\lambda(D+mL)$ is strictly monotone decreasing with respect to $m$. If $d_i=d$, for all $i$ then $\lambda(D+mL)=d$.
	
	(ii) As $m\rightarrow \infty$, we have
	$$\lim_{m\rightarrow \infty}\lambda(D+mL)=\frac{1}{N}\sum_{i=1}^N d_i.$$
\end{lemma}

(iii) $\lambda (D+mL)<\max_{1\leq i\leq N}d_i$
for all $m$.

\begin{proof}
	Since $L$ has non-negative non-diagonal elements, is irreducible, and satisfies (\ref{diagonal}), the 
	Perron-Frobenius theorem implies that its maximal eigenvalue 
	is $\lambda(L)=0$, with corresponding eigenvector $\bold{1}=(1,1,\cdots, 1)$. Note that $\lambda(L)=0$ means that $L$ is negative semi-definite.
	
	If $A,B$ are symmetric matrices and $B$ is positive semi-definite, then $\lambda(A+B)\geq \lambda(A)$ - this is a consequence of the
	min-max principle (see {\it{e.g.}}, \cite{bhatia}, Corrollary III.2.3). Apply this with $A=D+m_2L$, $B=(m_1-m_2)L$,
	where $0\leq m_1< m_2$ - note that $B$ is positive semi-definite since $L$ is negative semi-definite -
	to conclude that  $\lambda(D+m_2L)=\lambda(A)\leq \lambda(A+B)=\lambda(D+m_1L)$. We have therefore shown that
	$\lambda(D+mL)$ is monotone decreasing (the fact that it is strictly decreasing  when the $d_i$'s are not all equal will be shown below). 
	
	To prove (ii), note that $\lambda(D+mL)=m\cdot \lambda (L+m^{-1}D)$, and study 
	the maximal eigenvalue of $L+\epsilon D$ for small $\epsilon$. As $\epsilon\rightarrow 0+$, each of eigenvalues of $M(\epsilon)=L+\epsilon D$ converges to a corresponding eigenvalue of $L$ (see {\it{e.g.}} \cite{bhatia}, Corollary III.2.6), and since the maximal eigenvalue of $L$ is $0$, the maximal eigenvalue of 
	$M(\epsilon)$, which we denote by $\mu(\epsilon)$, 
	satisfies $\lim_{\epsilon\rightarrow 0}\mu(\epsilon)=0$.
	Denoting the eigenvector corresponding to $\mu(\epsilon)$ by $\bold{v}(\epsilon)$, we have $\bold{v}(0)=\bold{1}=(1,1,\cdots,1)$ and
	$$M(\epsilon)\bold{v}(\epsilon)=\mu(\epsilon) \bold{v}(\epsilon)\;\;\Rightarrow\;\;M'(\epsilon)\bold{v}(\epsilon)+M(\epsilon)\bold{v}'(\epsilon)=\mu'(\epsilon) \bold{v}(\epsilon) +\mu(\epsilon) \bold{v}'(\epsilon).$$
	Setting $\epsilon=0$ and denoting by $\langle \cdot,\cdot \rangle$ the standard inner product on $\Real^N$, we get
	$$ D\bold{1}+L\bold{v}'(0)=\mu'(0) \bold{1}\;\;\Rightarrow\;\; \langle D\bold{1},\bold{1} \rangle +\langle L\bold{v}'(0),\bold{1}\rangle=\mu'(0) \langle \bold{1},\bold{1}\rangle,$$
	and since by symmetry of $L$ we have $\langle L\bold{v}'(0),\bold{1}\rangle=\langle \bold{v}'(0),L\bold{1}\rangle=0$, we conclude that
	$$\mu'(0)=\frac{\langle D\bold{1},\bold{1} \rangle}{\langle \bold{1},\bold{1}\rangle}=\frac{1}{N}\sum_{i=1}^N d_i, \;\;\;\Rightarrow\;\;\;\mu(\epsilon)=\mu'(0)\epsilon +O(\epsilon^2)=\frac{\epsilon}{N}\sum_{i=1}^N d_i+O(\epsilon^2).$$
	Therefore, as $m\rightarrow \infty$,
	\begin{eqnarray*}\lambda(D+mL)&=&m\cdot \lambda(L+m^{-1}D)=m\cdot \left( \frac{m^{-1}}{N}\sum_{i=1}^N d_i+O(m^{-2})\right)\nonumber\\
		&=&\frac{1}{N}\sum_{i=1}^N d_i+O(m^{-1}),
	\end{eqnarray*}
	proving (ii).
	
	We now note that, since the function $m\rightarrow\lambda(D+mL)$ is real-analytic, it cannot be constant on an interval unless 
	it is everywhere constant, but if this is the case then part (ii) implies that
	$$\frac{1}{N}\sum_{i=1}^N d_i=\lim_{m\rightarrow\infty}\lambda(D+mL)=\lambda(D)=\max_{1\leq i\leq N} d_i,$$ which occurs 
	iff all $d_i$'s are equal. Therefore if the $d_i$'s are not equal we have
	that $m\rightarrow \lambda(D+mL)$ is strictly decreasing, and when the all $d_i$'s are equal, it is constant.
	
	(iii) follows from (i), since
	$\lambda(D+mL)\leq \lambda(D)=\max_{1\leq i\leq N}d_i$.
\end{proof}

In the following lemma we derive properties of the function $\Lambda_0(m)$,
which will be used in the proofs of the main theorems:

\begin{lemma}\label{prl}
	(i) $\Lambda_0(0)=\chi$, where $\chi$ is given by (\ref{dchi}).
	
	(ii) $\lim_{m\rightarrow\infty} \Lambda_0(m)=\bar{r}\doteq\frac{1}{N}\sum_{i=1}^N\bar{r}_i$.
	
	(iii) Assuming that it is not the case that $r_i(\theta)\equiv r_j(\theta)$ for all $i,j$,	the function $\Lambda_0(m)$ is strictly monotone decreasing on $[0,\infty)$.
	
	(iv) If $\chi<0$ then $\Lambda_0(m)<0$ for all
	$m>0$.
	
	(v) If $\chi>0$ and $\bar{r}<0$, there is a unique value $m^*$ such that $\Lambda_0(m^*)=0$,
	that is a solution of the equation (\ref{al}), and we have 
	\begin{eqnarray}\label{mv1}
		m\in (0,m^*)\;\;&\Rightarrow&\;\;\Lambda_0(m)>0,\nonumber\\
		m> m^*\;\;&\Rightarrow&\;\;\Lambda_0(m)<0.
	\end{eqnarray}
\end{lemma}

\begin{proof}
	(i) Since $R(\theta)$ is a diagonal matrix we have
	$\lambda(R(\theta))=r_{max}(\theta)$, hence
	by the definition (\ref{l0}) of $\Lambda_0(m)$, we have
	$$\Lambda_0(0)=\frac{1}{2\pi}\int_0^{2\pi}\lambda(R(\theta))d\theta=\frac{1}{2\pi}\int_0^{2\pi}r_{max}(\theta)d\theta=\chi.$$
	
	(ii) Fixing $\theta$ and applying Lemma \ref{matrixa}(ii) with $D=R(\theta)$ we have
	\begin{equation}\label{lime}\lim_{m\rightarrow \infty}\lambda (R(\theta)+mL)=\frac{1}{N}\sum_{i=1}^N r_i(\theta).\end{equation}
	Since, by Lemma \ref{matrixa}(i), $\lambda(R(\theta)+mL)$ is monotone 
	decreasing or constant with respect to $m$ for each fixed $\theta$, Lebesgue's Monotone Convergence Theorem (\cite{rudin}, Th. 11.28) and 
	(\ref{lime}) imply
	$$\lim_{m\rightarrow\infty}\Lambda_0(m)=\frac{1}{2\pi}\lim_{m\rightarrow\infty}\int_0^{2\pi}\lambda(R(\theta)+mL)d\theta=\frac{1}{2\pi}\frac{1}{N}\int_0^{2\pi}\sum_{i=1}^N r_i(\theta)d\theta=\bar{r}.$$
	
	(iii) By Lemma \ref{matrixa}(i) $\lambda(R(\theta)+mL)$ is monotone decreasing or constant with respect to $m$ for any value of theta $\theta$,
	and it is strictly monotone decreasing unless all
	$r_i(\theta)$'s are equal, so the definition (\ref{l0}) implies that $\Lambda_0(m)$ 
	is strictly monotone increasing unless all $r_i(\theta)$'s are everywhere equal.
	
	(iv) follows  from (i) and (iii).
	
	(v) If $\chi>0$ and $\bar{r}<0$ then by (i),(ii) we have $\Lambda_0(0)>0$, $\lim_{m\rightarrow \infty} \Lambda_0(m)<0$, and since by (iii) the function 
	$\Lambda_0(m)$ is strictly decreasing (note that 
	$\chi>0$, $\bar{r}<0$ preclude the possibility that all $r_i(\theta)$ are indentical) we conclude that there exists a unique $m^*>0$ with
	$\Lambda_0(m)=0$, and that (\ref{mv1}) holds.
\end{proof}

\subsection{The high frequency limit}
\label{high}
We now study the behavior of $\Lambda(m,\omega)$ in the limit 
$\omega\rightarrow\infty$.

Theorem 1.2 of \cite{liu} tells us that 
\begin{lemma}\label{oi}
	Let $\Lambda_\infty(m)$ be defined by (\ref{domi}).
	Then we have, for all $m>0$,
	$$\lim_{\omega\rightarrow\infty}\Lambda(m,\omega)=\Lambda_\infty(m).$$
\end{lemma}

We now derive properties of the function
$\Lambda_\infty(m)$, which will be used in the proofs of the main theorems.

\begin{lemma}\label{pli}
	(i) $\Lambda_\infty(0)=\max_{1\leq i\leq N}\bar{r}_i$.
	
	(ii) $\lim_{m\rightarrow \infty}\Lambda_{\infty}(m)=\bar{r}\doteq \frac{1}{N}\sum_{i=1}^N \bar{r}_i.$
	
	(iii) $\Lambda_\infty(m)$ is strictly monotone decreasing, except in the case that $\bar{r}_i=\bar{r}_j$ for all $i,j$, in which it is constant.
	
	(iv) For all $m>0$ we have
	$\Lambda_\infty(m)\leq \Lambda_0(m)$,
	with strict inequality unless all functions $r_i(\theta)-r_j(\theta)$ 
	are constant.
	
	(v) When $\bar{r}>0$ we have $\Lambda_\infty(m)>0$ for all $m>0$.
	
	(vi) When $\bar{r}_i<0$ for all $i$, we have 
	$\Lambda_\infty(m)<0$ for all $m>0$.
	
	(vii) When $\bar{r}<0$ and (\ref{pn}) holds, the 
	equation (\ref{li0}) has a unique solution
	$\hat{m}$, and 
	\begin{eqnarray}\label{mv3}
		m\in (0,\hat{m})\;\;\Rightarrow\;\; \Lambda_\infty(m)>0,\nonumber\\
		m>\hat{m}\;\;\Rightarrow\;\; \Lambda_\infty(m)<0.
	\end{eqnarray}
\end{lemma}
\begin{proof}
	
	(i) $\Lambda_\infty(0)=\lambda(\bar{R})=\max_{1\leq i\leq N}\bar{r}_i$.
	
	(ii) and (iii) follow from Lemma \ref{matrixa}(i),(ii) taking $D=\bar{R}$.
	
	(iv) follows by combining the results of Lemmas \ref{monotone}, \ref{omega0} and \ref{oi}: since 
	$\Lambda(m,\omega)$ is monotone decreasing in 
	$\omega$, 
	$$\Lambda_\infty(m)=\lim_{\omega\rightarrow\infty}\Lambda(m,\omega)\leq \lim_{\omega\rightarrow 0+}\Lambda(m,\omega)=\Lambda_0(m).$$
	Unless all $r_i(\theta)-r_j(\theta)$ are constant, Lemma 
	\ref{monotone} implies that the inequality is strict.
	
	(v) Follows from (ii) and (iii).
	
	(vi) Follows from (i) and (iii).
	
	(vii) If $\bar{r}<0$ and (\ref{pn}) holds, then
	by (i) we have $\Lambda_\infty(0)>0$,
	by (ii) we have $\lim_{m\rightarrow\infty}\Lambda_\infty(m)<0$, and
	by (iii) we have that $\Lambda_\infty(m)$ is strictly monotone decreasing. Hence there is 
	a unique value $\hat{m}\in (0,\infty)$ such that
	$\Lambda_\infty(\hat{m})=0$, and (\ref{mv3}) holds.
\end{proof}

\section{Proofs of the main theorems}
\label{mainproof}

We now combine the results obtained in the previous sections to obtain the proofs of the main theorems.

\begin{proof}[Proof of Theorem \ref{sinksink}]
	Here we assume the all-sink case, $\bar{r}_i<0$ for $1\leq i\leq N$.
	
	Let $\chi$ be defined by (\ref{dchi}). (\ref{limit})
	follows from Lemma \ref{omega0} and Lemma \ref{prl}(i).
	
	If $\chi<0$ then Lemma \ref{prl}(iv) and Lemma 
	\ref{omega0} imply that, for any $m>0$, we have $\Lambda(m,\omega)<0$ for 
	$\omega$ sufficiently small, hence Lemma \ref{monotone} implies that $\Lambda(m,\omega)<0$ for all
	$\omega$. Therefore we have part (I) of Theorem \ref{sinksink}.
	
	We now assume $\chi>0$. Lemma \ref{prl}(v) then implies that 
	there is a unique solution $m^*$ of the equation (\ref{al}), and 
	that (\ref{mv1}) holds.
	Also, by Lemma \ref{pli}(vi), we have that
	\begin{equation}\label{mv2}
		\Lambda_\infty(m)<0, \;\;\;m\in (0,\infty).
	\end{equation}
	By Lemma \ref{monotone}, we have that $\Lambda(m,\omega)$ is strictly monotone decreasing with respect to $\omega$ - indeed 
	it is impossible that $r_i(\theta)-r_j(\theta)$ are
	constant for all $i,j$, since this would imply that 
	one the functions $r_i(\theta)$ is larger
	than all the others, leading to $\chi=\bar{r}_i<0$,
	in contradiction with our assumption $\chi>0$.
	
	We first prove part (II)b, fixing $m> m^*$. (\ref{mv1}) and Lemma
	\ref{omega0} imply that $\Lambda(m,\omega)<0$ for 
	$\omega$ sufficiently small, and since
	$\Lambda(m,\omega)$ is monotone decreasing with respect to $\omega$, we have
	$\Lambda(m,\omega)<0$ for all $\omega>0$. 
	Note also that by continuity this implies $\Lambda(m^*,\omega)\leq 0$ for all $\omega$, and since
	$\Lambda(m^*,\omega)$ is strictly decreasing with respect to $\omega$ we conclude that $\Lambda(m^*,\omega)<0$ also holds.
	We thus have (II)(b).
	
	To prove part (II)(a) of the theorem, we now fix $m\in (0,m^*)$. (\ref{mv1}) and Lemma \ref{omega0}
	imply that $\Lambda(m,\omega)>0$ for $\omega>0$ sufficiently small, while (\ref{mv2}) and Lemma \ref{oi}
	imply that $\Lambda(m,\omega)<0$ for $\omega>0$ sufficiently large. By the monotonicity of $\Lambda(m,\omega)$ with respect to $\omega$, the above two facts imply that there exists a unique value
	$\omega=\omega_{c}(m)>0$ for which $\Lambda(m,\omega)=0$,
	and we have
	\begin{eqnarray}\label{bp}
		\omega \in (0,\omega_{c}(m))\;\;&\Rightarrow&\;\;
		\Lambda(m,\omega)>0\nonumber\\
		\omega> \omega_{c}(m)\;\;&\Rightarrow&\;\; \Lambda(m,\omega)<0.
	\end{eqnarray}
	Since the function $\Lambda(m,\omega)$ is real-analytic (see remark preceding Lemma \ref{connect}), and $\omega_c(m)$ is defined 
	implicitly by $\Lambda(m,\omega_c(m))=0$ (for $m\in (0,m^*)$), and $\Lambda_\omega'(m,\omega_c(m))<0$ (Lemma \ref{monotone}), the real-analytic implicit function theorem (see, {\it{e.g.}}, \cite{krantz}, Section 6.1) implies that
	$\omega_c(m)$ is a real-analytic function.
	
	To complete the proof we show that the function $\omega_c(m)$ can be continuously extended to the closed interval $[0,m^*]$,
	with $\omega_c(0)=\omega_c(m^*)=0$.
	To show that $\lim_{m\rightarrow 0+}\omega_c(m)=0$,
	we fix $\omega_0>0$ and show that $\omega_c(m)<\omega_0$
	for $m$ sufficiently small. Indeed, since $\bar{r}_i<0$ for all $i$, we know by (\ref{m0}) that for $m=0$ we have $\Lambda(0,\omega_0)=\max_{1\leq i\leq N}\bar{r}_i<0$, hence by continuity 
	$\Lambda(m,\omega_0)<0$ for $m$ sufficiently small, so that,
	for such $m$, (\ref{bp}) implies $\omega_c(m)<\omega_0$.
	Similarly, to show that $\lim_{m\rightarrow m^*-}\omega_c(m)=0$,
	we fix $\omega_0>0$ and show that $\omega_c(m)<\omega_0$
	for $m<m^*$ sufficiently close to $m^*$. 
	Indeed, by part (II)(b), proved above,
	we have $\Lambda(m^*,\omega_0)<0$, hence by continuity 
	$\Lambda(m,\omega_0)<0$ for $m$ sufficiently close to 
	$m^*$, so that, for such $m$, (\ref{bp}) implies $\omega_c(m)<\omega_0$.
\end{proof}

\begin{proof}[Proof of theorem \ref{sourcesink}]
	Here we assume the source-sink case (\ref{pn}).
	In the case $\bar{r}>0$, Lemmas \ref{monotone}, \ref{oi} and \ref{pli}(v) imply $\Lambda(m,w)\geq \lim_{\omega\rightarrow \infty} \Lambda(m,\omega)=\Lambda_\infty(m)>0$, proving part (I) of the theorem. 
	
	To prove part (II), we now assume $\bar{r}<0$.
	By Lemma \ref{prl}(i) and (\ref{pn}) we have
	\begin{eqnarray*}\Lambda_0(0)&=&\chi=\frac{1}{2\pi}\int_0^{2\pi}\max_{1\leq i\leq N}r_{i}(\theta)d\theta\geq \frac{1}{2\pi}\max_{1\leq i\leq N} \int_0^{2\pi}r_i(\theta)d\theta=\max_{1\leq i\leq N}\bar{r}_i>0\end{eqnarray*}
	Therefore Lemma \ref{prl}(v) implies that equation (\ref{al}) has a unique solution $m^*$, and that (\ref{mv1}) holds. 
	
	By Lemma \ref{pli}(vii) we have that a solution $\hat{m}$ of (\ref{li0}) exists, and (\ref{mv3}) holds. 
	Assume now that $r_i(\theta)-r_j(\theta)$ are not all constant.
	By Lemma \ref{pli}(iv) we have
	$$\Lambda_0(m^*)=0=\Lambda_\infty(\hat{m})< \Lambda_0 (\hat{m}),$$
	which, since $\Lambda_0(m)$ is a decreasing function (Lemma \ref{prl}(iii)), implies $\hat{m}< m^*$.
	
	We consider three cases:
	
	a. Assume $m\in (0,\hat{m})$. Then (\ref{mv3}) and Lemma \ref{oi} imply that $\Lambda(m,\omega)>0$ for $\omega$ sufficiently large.
	But since $\Lambda(m,\omega)$ is decreasing with respect to $\omega$ (Lemma \ref{monotone}) we conclude that $\Lambda(m,\omega)>0$ for
	all $\omega>0$. By continuity we obtain also $\Lambda(\hat{m},\omega)\geq 0$
	for all $\omega>0$, and, since $ \Lambda(m^*,\omega)$ is strictly decreasing with respect to $\omega$, this implies $\Lambda(\hat{m},\omega)> 0$.
	
	b. Assume that $m\in (\hat{m},m^*)$. Since $m<m^*$, (\ref{mv1}) and Lemma \ref{omega0} imply that $\Lambda(m,\omega)>0$ for $\omega>0$ sufficiently small. On the other hand (\ref{mv3}) and Lemma \ref{oi} imply $\Lambda(m,\omega)<0$
	for $\omega$ sufficiently large. These two facts, together with the strict monotonicity of $\Lambda(m,\omega)$ with respect to $\omega$, imply that there is a unique value $\omega_c(m)$ such that (\ref{bp}) holds, which is the desired conclusion. 
	
	c. Assume $m>m^*$. Then (\ref{mv1}) and
	Lemma \ref{omega0} imply that $\Lambda(m,\omega)<0$ for $\omega>0$ sufficiently small, and since
	$\Lambda(m,\omega)$ is monotone decreasing with resepct to $\omega$ we conclude that $\Lambda(m,\omega)<0$ for all $\omega>0$. By continuity 
	this implies $\Lambda(m^*,\omega)\leq 0$ for all $\omega$, which by the strict
	monotonicity of $\Lambda(m^*,\omega)$ with respect to $\omega$ implies $\Lambda(m^*,\omega)<0$.
	
	To conclude, we prove the properties of the function $\omega_c(m)$ stated in the theorem. Real-analyticity follows from the implicit function theorem, as 
	in the proof of Theorem \ref{sinksink} above. 
	
	To show that 
	$\lim_{\omega\rightarrow \hat{m}+}\omega_c(m)=+\infty$, we fix $\omega_0>0$.
	By case a above, we have $\Lambda(\hat{m},\omega_0)>0$, hence 
	by continuity $\Lambda(m,\omega_0)>0$ for $m>\hat{m}$ sufficiently close to 
	$\hat{m}$, which, by (\ref{bp}), implies $\omega_c(m)>\omega_0$ for such $m$.
	
	To show that 
	$\lim_{\omega\rightarrow m^*-}\omega_c(m)=0$, we fix $\omega_0$.
	By case c above, we have $\Lambda(m^*,\omega_0)<0$, hence $\Lambda(m,\omega_0)<0$
	for $m<m^*$ sufficiently close to $m^*$, which, by (\ref{bp}), implies $\omega_c(m)<\omega_0$ for such $m$.

\end{proof}

\section{Discussion}

The DIG effect is an interesting example of an emergent dynamical phenomenon
which arises from the combination of several elementary mechanisms, and which
cannot occur if any of the mechanisms is excluded.
The mechanisms here are: (i) Temporal heterogeneity: population growth rate of at least one patch varies  in time.
(ii) Spatial heterogeneity: the population growth rate profiles in the patches are not identical, (iii) Dispersal among the patches. 
Given these mechanisms we have seen that 
the populations can persist and grow despite the fact that each of the patches is a sink. In the absence of any one of these three mechanisms, population growth could not occur when all patches are sinks.

We have proved that DIG is a {\it{robust}} phenomenon, as it occurs {\it{regardless}} of the specific choice of the periodic local growth-rate profiles $r_i(\theta)$,  as long as the condition $\chi>0$ holds (with $\chi$ given by 
(\ref{dchi})), for 
a range of values of the frequency $\omega$ of the 
oscillations in growth rates and of the dispersal rate $m$. However we have seen that in order for 
DIG to occur, it is necessary that the frequency $\omega$ not be too large, and that the
disperal rate $m$ is neither too small
{\it{nor}} too large. 

Theorem \ref{sinksink} explains the main
features of the subset of parameters in the $(m,\omega)$ plane for which
the DIG pheonomenon occurs, as observed in the
numerical results presented in
Figure \ref{psinksink}, and discussed in Section
\ref{numerical}. An additional feature observed in this figure, and in analogous figures we have plotted 
for other periodic profiles $r_1(\theta),r_2(\theta)$, 
is that the function $\omega=\omega_c(m)$ is convex. We conjecture that this fact holds generally, and it is an interesting challenge to 
prove this.

Another numerical observation concerns 
the behavior of the curve $\omega=\omega_{c}(m)$ in the
vicinity of $m=0$. From the
numerical results it seems evident that this curve is tangent to the $\omega$-axis at the origin, which leads to the conjecture that
$\lim_{m\rightarrow 0}\omega_{c}'(m)=+\infty$. 
Note that this means that for
small $\omega$ very weak dispersal is sufficient to cause DIG. 
We note that studying this case in which $m$ and $\omega$ are 
simultaneously small is rather delicate.
We raised the above conjecture in an earlier preprint version of this work,
and it has now been established, in the case $N=2$ patches, in \cite{benaim} for piecewise-constant 
profiles $r_i(\theta)$ using explicit computations and for general profiles
in \cite{lobry}, using techniques of nonstandard analysis.

While the results obtained here show that the region 
in the $(m,\omega)$ parameter plane for which 
DIG occurs has qualitative features which are
independent of the topology of the network 
of patches, as encoded in the matrix $L$, 
it is of interest to furter explore 
how the quantitative properties of this 
set of parameters depends on the network topology, as well as on the 
form of the period growth profiles.
Such improved understanding 
would enable to better assess the extent and the circumstances under which the DIG effect is
relevant to explaining population persistence and 
growth in real-world ecosystems.



{}

\end{document}